\documentclass[preprint,11pt]{elsarticle}
\makeatletter
\def\ps@pprintTitle{%
 \let\@oddhead\@empty
 \let\@evenhead\@empty
 \def\@oddfoot{}%
 \let\@evenfoot\@oddfoot}
\makeatother

\usepackage[titletoc,title]{appendix}

\makeatletter
  \@addtoreset{chapter}{part}
  \@addtoreset{@ppsaveapp}{part}
\makeatother

\usepackage[utf8]{inputenc}

\usepackage{graphicx} 
\usepackage{enumitem}
\usepackage{amsmath}
\usepackage{amssymb}
\usepackage{amsthm}
\usepackage{amsmath,amsthm}
\usepackage{mathtools}
\usepackage{algorithm}
\usepackage{algcompatible}
 \usepackage{setspace}
 \usepackage{xspace} 
 
\floatname{algorithm}{}
\let\Algorithm\algorithm
\renewcommand\algorithm[1][]{\Algorithm[#1]\setstretch{1.2}}

\newcommand{\Algo}{\textit{LPT-REV}\xspace}
\newcommand{\Slack}{SLACK\xspace}
\newtheorem{theorem}{Theorem}
\newtheorem{proposition}{Proposition}

\newtheorem{lemma}{Lemma}

\begin{document}
\pagestyle{plain}

\begin{frontmatter}

\title{Longest Processing Time rule for identical parallel machines revisited}

 \author[1,2]{Federico Della Croce\corref{cor1}}
 \author[1]{Rosario Scatamacchia}
\cortext[cor1]{Corresponding author.}
 
  \address[1]{\small Dipartimento di Ingegneria Gestionale e della Produzione, Politecnico di Torino,\\ Corso Duca degli Abruzzi 24, 10129 Torino, Italy, \\{\tt \{federico.dellacroce, rosario.scatamacchia\}@polito.it }}
 \address[2]{CNR, IEIIT, Torino, Italy}

\begin{abstract}
We consider the $P_m || C_{max}$ scheduling problem where the goal is to schedule $n$ jobs on $m$ identical parallel machines to minimize makespan. We revisit the famous Longest Processing Time ($LPT$) rule proposed by Graham in 1969. $LPT$ requires to sort jobs in non-ascending order of processing times and then to assign one job at a time to the machine whose load is smallest so far. We provide new insights on LPT and discuss the approximation ratio of a modification of $LPT$ that improves Graham's bound from $\left( \frac{4}{3} - \frac{1}{3m} \right)$ to $\left( \frac{4}{3} - \frac{1}{3(m-1)} \right)$ for $m \geq 3$ and from $\frac{7}{6}$ to $\frac{9}{8}$ for $m=2$. We use Linear Programming (LP) to analyze the approximation ratio of our approach. This performance analysis can be seen as a valid alternative to formal proofs based on analytical derivation. Also, we derive from the proposed approach an $O(n \log n)$ heuristic. The heuristic splits the sorted jobset in tuples of $m$ consecutive jobs ($1,\dots,m; m+1,\dots,2m;$ etc.) and sorts the tuples in non-increasing order of the difference (slack) between largest job and smallest job in the tuple. Then, List Scheduling is applied to the set of sorted tuples. This approach strongly outperforms $LPT$ on benchmark literature instances.
\end{abstract}

\end{frontmatter}

\section{Introduction}
\label{sec:intro}

We consider the $P_m || C_{max}$ scheduling problem (as denoted in the three--field classification by \cite{GRAHAM1979}) where the goal is to schedule $n$ jobs on $m$ identical parallel machines $M_i~(i=1,\dots,m)$ to minimize the makespan.  $P_m || C_{max}$ is strongly NP-hard (\cite{GaJo79}) and has been intensively investigated in the literature both from a theoretical and a practical point of view. For an exhaustive discussion we refer, among others,  to books \cite{Leung04,Pinedo16} and to the comprehensive survey \cite{ChPoWo99}. The 
pioneering approximation algorithm for the problem is the Longest Processing Time (LPT) rule proposed in \cite{Graham69}. It requires to sort the jobs in non-ascending order of their processing times $p_j ~(j=1,\dots,n)$ and then to assign one job at a time to the machine whose load is smallest so far. This assignment of jobs to machines is also known as List Scheduling (LS). Several properties have been established for LPT in the last decades \cite{Graham69, CofSe76, BoChen93, BloSe15}. We recall the main theoretical results for LPT in the next section. LPT generally exhibits much better performance in practice than the expected theoretical ratios, especially as the number of jobs gets larger.  In \cite{FrKan87}, it was also showed that LPT is asymptotically optimal under mild assumptions on the input data. Due to its simplicity and practical effectiveness, LPT became a cornerstone for the design of more involving exact or heuristic algorithms. \\
We mention other popular approximation algorithms which exploit connections of $P_m || C_{max}$ with bin packing: \textit{Multifit} \cite{CofGaJo78}, \textit{Combine} \cite{LeeMa88} and \textit{Listfit} \cite{GuTo01}. Such algorithms provide better worst case performance than LPT but at the cost of higher running times. Also, Polynomial Time Approximation Schemes (PTASs) were derived for the problem. The first PTAS was given in \cite{HoSh87}. PTASs with improved running times were then provided in \cite{AlAzWoYa98, Hochbaum97, Jansen10}. Recently, an improved PTAS has been proposed in \cite{JaKlVe17}. \\

The contribution of this work is twofold. First, we revisit the LPT rule and provide a simple algorithmic variant that manages to improve the longstanding Graham's approximation ratio derived in \cite{Graham69} keeping the same computational complexity. 
To establish our theoretical results, we also employ Linear Programming (LP) to analyze the worst case performance of the proposed algorithm and to derive approximation bounds. 
In a sense, this paper can also be seen as a followup of the work in \cite{MOV97} where several LPs where used to determine
the worst case approximation ratio of LPT on two uniform machines. 
Recently a growing attention has been paid to the use of LP modeling for the derivation of formal proofs (see \cite{CW16,ACCEHMW16,DEPFSC17}) and we also show here a successful application of this technique.\\
We then move from approximation to heuristics. By generalizing the proposed LPT--based approach, we obtain a simple algorithm 
running in $\mathcal{O}(n\log n)$ time which drastically improves upon LPT and can hence be regarded as a valuable alternative to the most popular constructive heuristic designed for this problem. 

\section{Notation and LPT properties}
\label{sec:LPTreview}

We first recall the main theoretical properties of LPT applied to $P_m || C_{max}$.  
From now on, we will consider the jobs sorted by non-increasing $p_j$ ($p_j \geq p_{j+1}, \;\; j=1,\dots,n-1$). We denote the solution values of the LPT schedule and the optimal makespan by $C_{m}^{LPT}$ and $C_{m}^*$ respectively, where index $m$ indicates the number of machines. Also, we denote by $r^{LPT} = \frac{C_{m}^{LPT}}{C_{m}^*}$ the performance ratio of LPT and, as in \cite{BoChen93}, by $r_k^{LPT}$ 
 the performance ratio of an LPT schedule with $k$ jobs assigned to the machine yielding the completion time (the critical machine) and by $j^\prime$ the job giving the corresponding makespan (the critical job). We summarize hereafter bounds on  $C_m^*$ and properties of the LPT schedule available in the literature.
\begin{proposition} \label{LitProp1} 
\cite{Pinedo16} The following expressions hold:
\begin{align}
& C_m^* \geq \max \left \{ \frac{\sum\limits_{j=1}^{n} p_j}{m}, p_1 \right \}; \label{TheOneOpt}\\
& C_{m}^{LPT} = C_m^* \text{ if }  p_{j^\prime} > \frac{C_{m}^*}{3};\\
& C_{m}^{LPT} \leq  \frac{\sum\limits_{j=1}^{j^\prime} p_j}{m} +  p_{j^\prime}(1-\frac{1}{m}) \leq C_{m}^* + p_{j^\prime}(1-\frac{1}{m}). \label{GenRel}
\end{align}
\end{proposition}

\begin{proposition} \label{BChenProp}
\cite{BoChen93} For each job $i$ assigned by LPT in position $j$ on a machine, the following inequality holds
\begin{align}
& p_i \leq \frac{C_{m}^*}{j}. \label{ChenRes} 
\end{align}
\end{proposition}

\begin{proposition} \label{LitBounds}
The following tight approximation ratios hold for LPT:
\begin{align}
& r_2 \leq \frac{4}{3} - \frac{1}{3(m-1)}; \quad \text{\cite{BoChen93}} \label{Chen2m} \\
& r_k  \leq \frac{k+1}{k} - \frac{1}{km} \quad k \geq 3. \quad \text{\cite{CofSe76}} \label{GenBound}
\end{align}
\end{proposition}
Approximation bounds \eqref{GenBound} were derived in \cite{CofSe76} and are also known as the a-posteriori generalized bounds. For $k=3$, the corresponding approximation ratio of $\frac{4}{3} - \frac{1}{3m}$ is the well--known Graham's bound \cite{Graham69} and constitutes the worst case bound for LPT. A straightforward implication of Property \eqref{ChenRes} is the following: 
\begin{lemma}
\label{BoChenImplied} 
If LPT provides a schedule where a non critical machine processes at least $k$ jobs before the critical job $j^\prime$, then $r^{LPT}\leq \frac{k+1}{k} - \frac{1}{km}$.
\end{lemma}
\begin{proof}
Denote by $i$ the job in the $k-th$ position on the non critical machine, with $p_i \geq p_{j^\prime}$. Since we have $p_i \leq \frac{C_{m}^*}{k}$ from Property \eqref{ChenRes} and thus $p_{j^\prime} \leq \frac{C_{m}^*}{k}$ also holds, we get from inequality \eqref{GenRel} that $C_{m}^{LPT} \leq (\frac{k+1}{k} - \frac{1}{km} ) C_{m}^*$.
\end{proof}


\section{LPT revisited}
\label{LPTrevisiting}
We provide here further insights on the Longest Processing Time rule. As usually employed in the worst case performance analysis of LPT, we assume that the critical job is the last one, namely $j^\prime = n$. Otherwise, we would have other jobs scheduled after the critical job that do not affect the makespan provided by LPT but can contribute to increase the optimal solution value. 
\subsection{Results for LPT}
We first elaborate on the approximation bound provided in Lemma \ref{BoChenImplied}. We state the following proposition.
\begin{proposition}
\label{kJobsMac}
If LPT schedules at least $k$ jobs on a non critical machine before assigning the critical job, then $r^{LPT} \leq \frac{k+1}{k} - \frac{1}{k(m-1)}$ for $m \geq k + 2$. 
\end{proposition}
\begin{proof}
First, we can assume that the critical machine processes at least two jobs, otherwise the LPT solution would be optimal as $C_{m}^* \geq p_1$ due to Proposition \ref{LitProp1}. 
Also, due to Lemma \ref{BoChenImplied}, condition $p_{n} \leq \frac{C_{m}^*}{k}$ holds. Denote by $t_{c}$ the completion time of the critical machine before loading critical job $n$. We have $C_{m}^{LPT} = t_{c} + p_{n}$. Also, denote by $t^{\prime}$ the completion time of a non-critical machine processing at least $k$ jobs and by $t^{\prime\prime}$ the sum of completion times of the other $(m-2)$ machines, namely $t^{\prime\prime} = \sum\limits_{j=1}^{n} p_j - t^{\prime} - (t_{c} + p_{n})$. Since the application of List Scheduling to the sorted jobs, each of the $(m-2)$ machines must have a completion time greater than $t_{c}$. Hence, the following inequality holds 
\begin{align}
&\frac{t^{\prime\prime}}{m-2} \geq t_{c}.  \label{AverL3}
\end{align}
We now rely on Linear Programming to evaluate the worst case performance ratio $\frac{C_{m}^{LPT}}{C_{m}^*}$. More precisely, we introduce an LP formulation where we can arbitrarily set the value $C_{m}^{LPT}$ to 1 and minimize the value of $C_{m}^*$. We associate non-negative variables $sum_p$ and $opt$ with $\sum\limits_{j=1}^{n} p_j$ and $C_{m}^*$, respectively. We also consider completion times $t_{c}$, $t^{\prime}$, $t^{\prime\prime}$ and processing time $p_n$ as non-negative variables in the LP model. Since we have $p_{n} \leq \frac{C_{m}^*}{k}$, we introduce an auxiliary slack variable $sl$ to write the corresponding constraint in the LP model as $p_{n} + sl - \frac{opt}{k} = 0$. The following LP model is implied: 
{\small
\begin{align}
   \text{minimize}\quad & opt \label{eq:ObjOpt}\\
	\text{subject to}\quad
	& - m\cdot opt + sum_p  \leq 0  \label{eq:BoundOpt}\\
	& k\cdot p_n - t^{\prime} \leq 0  \label{eq:l1}\\
	& t_{c} - t^{\prime}  \leq 0 \label{eq:L1t1}\\
    & (m-2)t_{c}  - t^{\prime\prime}  \leq 0  \label{eq:avL3}\\
    & (t_{c}  + p_n) + t^{\prime} + t^{\prime\prime}  - sum_p = 0 \label{eq:sump}\\	
    & t_{c} + p_n = 1 \label{eq:heur} \\
	& p_n+ sl - \frac{opt}{k}  = 0  \label{eq:pnOpt3} \\
	& t_{c}, t^{\prime}, t^{\prime\prime}, p_n, sum_p, opt, sl \geq 0  \label{eq:TheVar} 
\end{align}}
The minimization of the objective function \eqref{eq:ObjOpt}, after setting without loss of generality the LPT solution value to 1 (constraint \eqref{eq:heur}), provides an upper bound on the performance ratio of LPT rule. Constraint  \eqref{eq:BoundOpt} represents the bound $C_{m}^* \geq \frac{\sum\limits_{j=1}^{n} p_j}{m}$ while constraint \eqref{eq:l1} states that the value of $t^{\prime}$ is at the least $kp_{n}$, since $k$ jobs with greater processing than $p_{n}$ are assigned to the non critical machine. Constraint  \eqref{eq:L1t1} states that the completion time of the critical machine before the execution of the last job is not superior to the completion time of the other machine processing at least $k$ jobs. Constraint \eqref{eq:avL3} fulfills inequality \eqref{AverL3}. Constraint \eqref{eq:sump} guarantees that variable $sum_p$ corresponds to $\sum\limits_{j=1}^{n} p_j$ and constraint \eqref{eq:pnOpt3} represents condition $p_{n} \leq \frac{C_{m}^*}{k}$. Eventually, constraints \eqref{eq:TheVar} state that all variables are non-negative. A feasible solution of model \eqref{eq:ObjOpt}--\eqref{eq:TheVar} for any value of $m$ is:
\begin{align}
 & t_{c} = \frac{k(m - 1) - 1}{(k + 1)m - k - 2}; ~  p_n = \frac{m - 1}{(k + 1)m - k - 2}; \nonumber \\
 & t^{\prime} = \frac{k(m - 1)}{(k + 1)m - k - 2}; ~ t^{\prime\prime} = \frac{(m - 2)(k(m - 1) - 1)}{(k + 1)m - k - 2}; \nonumber \\
 & opt = \frac{k(m - 1)}{(k + 1)m - k - 2}; ~ sum_p = \frac{m(m - 1)k}{(k + 1)m - k - 2}; ~ sl = 0. \nonumber 
\end{align}
We can show by strong duality that such a solution is in fact optimal for any $m \geq k + 2$. Plugging $ p_n = \frac{opt}{k} - sl$ from constraint \eqref{eq:pnOpt3} in constraints \eqref{eq:sump} and \eqref{eq:heur}, we get an equivalent reduced LP model. If we associate dual variables $\lambda_i$ $(i=1,\dots,6)$ with constraints \eqref{eq:BoundOpt}--\eqref{eq:heur} respectively, the corresponding dual formulation of the reduced problem is as follows: 
{\small
\begin{align}
   \text{maximize}\quad & \lambda_6 \label{eq:ObjDual}\\
	\text{subject to}\quad
	& - m\lambda_1 + \lambda_2 + \frac{\lambda_5 + \lambda_6}{k} \leq 1  \label{eq:Dual1}\\
	& \lambda_1 - \lambda_5 \leq 0  \label{eq:Dual2}\\
	& -k\cdot \lambda_2 - \lambda_5 - \lambda_6 \leq 0   \label{eq:Dual3}\\
	& -\lambda_2 - \lambda_3 + \lambda_5 \leq 0   \label{eq:Dual4}\\
	& \lambda_3  + (m - 2)\lambda_4 + \lambda_5 + \lambda_6 \leq 0  \label{eq:Dual5}\\
	& -\lambda_4 + \lambda_5 \leq 0   \label{eq:Dual6}\\
	& \lambda_1, \lambda_2, \lambda_3, \lambda_4 \leq 0  \label{eq:TheDualVar}
\end{align}}
where the dual constraints \eqref{eq:Dual1}--\eqref{eq:Dual6} are related to the primal variables $opt$, $sum_p, sl$, $t^{\prime}, t_{c}, t^{\prime\prime}$ respectively. For any $m \geq k + 2$, a feasible solution of model \eqref{eq:ObjDual}--\eqref{eq:TheDualVar} is
\begin{align}
 & \lambda_1 =  \lambda_2 = \lambda_4 = \lambda_5 = \frac{-k}{(k + 1)m - k - 2}; ~ \lambda_3 = 0; ~ \lambda_6 = \frac{k(m - 1)}{(k + 1)m - k - 2}; \nonumber
\end{align}
where condition $m \geq k + 2$ is necessary to satisfy constraint \eqref{eq:Dual3}. Since $opt = \lambda_6$ in the above solutions, by strong duality these solutions are both optimal. We hence have
$$\frac{C_{m}^{LPT}}{C_{m}^*} \leq \frac{1}{opt} = \frac{(k + 1)m - k - 2}{k(m - 1)} = \frac{(k +1)(m - 1) - 1}{k(m - 1)} = \frac{k + 1}{k} - \frac{1}{k(m-1)}$$
which shows the claim. 
\end{proof}
Notably, with respect to Lemma \ref{BoChenImplied}, the result of Proposition \ref{kJobsMac} for $k = 3$ provides already a better bound than Graham's bound and equal to $\frac{4}{3} - \frac{1}{3(m-1)}$ for $m \geq 5$. Also, with the results of Proposition \ref{kJobsMac} we can state the following result. 
\begin{proposition}
\label{LPT2m+2}
In  $P_m || C_{max}$ instances with $n \geq 2m + 2$ and $m \geq 5$, $LPT$ has an approximation ratio $\leq \left(\frac{4}{3} - \frac{1}{3(m-1)}\right)$.
\end{proposition}
\begin{proof}
In instances with $n \geq 3m+1$, the claim straightforwardly holds since there exists at least one machine in the LPT schedule executing at least four jobs and so either bound $r_k$ or Lemma \ref{BoChenImplied}, with $k \geq 4$ respectively, applies.\\ 
Consider the remaining cases with $2m + 2 \leq n \leq 3m$.
We assume the critical job $n$ in third position on a machine, otherwise either bound $r_2$ holds or at least bound $r_4$ holds. This implies that $LPT$ schedules at least another job in position $\geq 3$ on a non critical machine. Hence, the results of Proposition \ref{kJobsMac} with $k = 3$ apply.
\end{proof}

In instances with $n \geq 2m+2$ and $3 \leq m \leq 4$, we can combine the reasoning underlying model \eqref{eq:ObjOpt}--\eqref{eq:TheVar} with a partial enumeration of the optimal/LPT solutions and state the following result.
\begin{proposition}
\label{mthfo}
In  $P_m || C_{max}$ instances with $n \geq 2m+2$, $LPT$ has an approximation ratio 
$\leq \left(\frac{4}{3} - \frac{1}{3(m-1)}\right)$ for $3 \leq m \leq 4$.
\end{proposition}
\begin{proof}
See Appendix. 
\end{proof}
Consider now instances with $2m$ jobs at most. The following proposition holds.
\begin{proposition}
\label{LPT2m}
In $P_m || C_{max}$ instances with $n \leq 2m$, 
$LPT$ has an approximation ratio $\leq \left(\frac{4}{3} - \frac{1}{3(m-1)}\right)$.
\end{proposition} 
\begin{proof}
We denote by  $C_{m}^{LPT}(\mathcal{J})$ the makespan given by LPT on jobset $\mathcal{J} = \{1,\dots, n\}$, with the jobs ordered by non-increasing $p_j$ $(j=1,\dots,n)$. We consider the case $n=2m$ only. All other cases $n < 2m$ can be reduced to the previous one after adding $2m - n$ dummy jobs with null processing time.\\
It is well known (see, e.g., \cite{Graham69}) that if each machine processes two jobs at most in an optimal schedule, the solution provided by LPT would be optimal. Hence, we consider the case where there is one machine processing at least three jobs in an optimal solution. This situation straightforwardly implies that job $1$ has to be processed alone on a machine. Therefore, we have $C_{m}^*(\mathcal{J}) \geq C_{m-1}^*(\mathcal{J} \setminus \{1\})$ since the optimal makespan with $m$ machines could be as well given by the machine processing only job 1.\\
On the other hand, to contradict the claim, LPT must have the critical machine processing more than two jobs, otherwise we could use the bound of Property \eqref{Chen2m}. This implies that job $1$ is processed alone on a machine and cannot give the makespan, otherwise LPT solution would be optimal due to Property \eqref{TheOneOpt}. We thus have $C_{m}^{LPT}(\mathcal{J}) = C_{m-1}^{LPT}(\mathcal{J} \setminus \{1\})$. Combining these results with Graham's bound on the problem instance with $m -1$ machines and without job 1, we get
\begin{align}
& \frac{C_{m}^{LPT}(\mathcal{J})}{C_{m}^*(\mathcal{J})} \leq \frac{C_{(m-1)}^{LPT}(\mathcal{J} \setminus \{1\})}{C_{(m-1)}^*(\mathcal{J} \setminus \{1\})} \leq \frac{4}{3} - \frac{1}{3(m-1)}. \nonumber
\end{align}
\end{proof}
For instances with exactly $2m + 1$ jobs, we provide the following proposition.
\begin{proposition}
\label{LPT2m+1}
In instances with $n = 2m + 1$, if LPT loads at least three jobs on a machine before the critical job, then the approximation ratio is not superior to $\frac{4}{3} - \frac{1}{3(m-1)}$. 
\end{proposition}
\begin{proof}
If LPT schedules at least three jobs on a machine before critical job $n$, this means that job 1 is processed either alone on a machine or with critical job $n$ only.  In the latter case, the claim is showed through the bound of Property \eqref{Chen2m}. Alternatively, job 1 is processed alone on machine $M_1$. Also, $M_1$ cannot give the makespan, otherwise LPT would yield an optimal solution. This implies that $C_{m}^{LPT}(\mathcal{J}) = C_{m-1}^{LPT}(\mathcal{J} \setminus \{1\})$ and that a trivial upper bound on the LPT solution value is equal to $p_1 + p_n$. In this case, if an optimal solution schedules job 1 with another job, we have $C_{m}^* \geq p_1 + p_n$ and thus LPT also gives the optimal makespan or else a contradiction on the optimal solution would occur. If an optimal solution schedules job 1 alone on a machine, then inequality $C_{m}^*(\mathcal{J}) \geq C_{m-1}^*(\mathcal{J} \setminus \{1\})$ holds. Combining these results with Graham's bound as in Proposition \ref{LPT2m}, we have 
\begin{align}
\frac{C_{m}^{LPT}(\mathcal{J})}{C_{m}^*(\mathcal{J})} \leq \frac{C_{(m-1)}^{LPT}(\mathcal{J} \setminus \{1\})}{C_{(m-1)}^*(\mathcal{J} \setminus \{1\})} \leq \frac{4}{3} - \frac{1}{3(m-1)}. \nonumber
\end{align}
\end{proof}
Summarizing, we have shown that LPT has an approximation ratio not superior to $\frac{4}{3} - \frac{1}{3(m-1)}$ in all instances with $m \geq 3$ and $n \neq 2m + 1$. Also, LPT can actually hit Graham's bound only in instances with $2m + 1$ jobs where only the critical machine processes three jobs and all the other machines process two jobs.

\subsection{Improving the LPT bound: Algorithm \Algo} \label{BeatingGraham}
We consider a slight algorithmic variation of LPT where a subset of the sorted jobs is first loaded on a machine and then LPT is applied to the remaining jobset. We denote this variant as $LPT(\mathcal{S})$ where $\mathcal{S}$ represents the set of jobs assigned alltogether to a machine first. 
Consider the following procedure.
\begin{algorithm}[H]
\begin{algorithmic}[1]
\STATEx \textbf{Input:} $P_m || C_{max}$ instance with $n$ jobs and $m$ machines.
\STATE Apply LPT yielding a schedule with makespan $z_1$ and $k-1$ jobs on the critical machine before job $j^\prime$.
\STATE Apply $LPT^\prime = LPT(\{j^\prime\})$ with solution value $z_2$.
\STATE Apply $LPT^{\prime\prime} = LPT(\{(j^\prime - k + 1), \dots, j^\prime\})$ with solution value $z_3$.
\STATE Return $\min\{z_1, z_2, z_3\}$.
\end{algorithmic}
\caption{\textbf{Algorithm \Algo}}
\end{algorithm}
In practice, \Algo algorithm applies $LPT$ first and then re-applies $LPT$ \textit{after} having loaded on a machine first either its critical job $j^\prime$ alone or the tuple of $k$ jobs $(j^\prime - k+1),...,j^\prime$. \\
In the following we will show that algorithm \Algo improves the longstanding Graham's bound from $\frac{4}{3} - \frac{1}{3m}$ to $\frac{4}{3} - \frac{1}{3(m-1)}$ for $m \geq 3$. 
For the performance analysis of algorithm \Algo, before addressing the remaining instances with $n = 2m+1$, we claim that the critical job in any LPT schedule can be again assumed to be $n$. Consider instances where there is a set of jobs loaded after the critical job $j^\prime$. If one of these jobs in not critical in either $LPT^\prime$ or $LPT^{\prime\prime}$ schedule, our claim would be already showed since further jobs after $j^\prime$ can only increase the optimal makespan without affecting the solution value of \Algo. Alternatively, the following proposition holds.
\begin{proposition}
\label{OtherJobs}
In $P_m || C_{max}$ instances where there are jobs processed after the critical job in the LPT solution and one of such jobs is critical in either $LPT^\prime$ or $LPT^{\prime\prime}$ schedules, \Algo algorithm has a performance guarantee of $\frac{4}{3} - \frac{7m - 4}{3(3m^2 + m - 1)}$. 
\end{proposition}
\begin{proof}
Denote by $\beta \sum\limits_{j=1}^{n} p_j$ the overall processing time of jobs $j^\prime + 1, \dots, n$, with $0 < \beta < 1$. Due to Graham's bound, the following relation holds for LPT when only jobs $1, \dots, j^\prime$ are considered:
\begin{align}
\label{GrahamInd}
&C_{m}^{LPT} \leq \frac{\sum\limits_{j=1}^{j^\prime} p_j}{m}\left(\frac{4}{3} - \frac{1}{3m}\right) = \frac{(1 - \beta)\sum\limits_{j=1}^{n} p_j}{m}\left(\frac{4}{3} - \frac{1}{3m}\right)
\end{align}
From \eqref{GrahamInd} we have:
\begin{align}
\label{LPTRHO}
&\frac{C_{m}^{LPT}}{C_{m}^*} \leq \frac{C_{m}^{LPT}}{\frac{\sum\limits_{j=1}^{n} p_j}{m}} \leq (1 - \beta)\left(\frac{4}{3} - \frac{1}{3m}\right) 
\end{align}
We introduce a target LPT approximation ratio denoted as $\rho$ and identify the value of $\beta$ which gives such a bound. We have:
\begin{align}
\label{BetaVal}
& (1 - \beta)\left(\frac{4}{3} - \frac{1}{3m}\right) = \rho \implies  \beta = 1 - \frac{3m\rho}{4m - 1}
\end{align}
Consider now the solution provided by $LPT^\prime$. Denote by $i$ $(j^\prime + 1 \leq i \leq n)$ the corresponding critical job and by $t_{c^\prime}$ the processing time of the remaining jobs on the critical machine. Since the following relations hold
\begin{align}
& t_{c^\prime} + \frac{p_i}{m} \leq C_{m}^{*}; \nonumber \\ 
& p_i \leq \beta \sum\limits_{j=1}^{n} p_j \leq  m\beta C_{m}^{*}; \nonumber
\end{align}
we have, in combination with  \eqref{BetaVal}, that
\begin{align}
\label{LSprop}
& C_{m}^{LPT^\prime} = t_{c^\prime} + p_i = \left(t_{c^\prime} + \frac{p_i}{m}\right) + p_i\left(1 - \frac{1}{m}\right)\leq C_{m}^* +(m-1) \beta C_{m}^*  \nonumber \\ 
& \implies  \frac{C_{m}^{LPT^\prime}}{ C_{m}^*} \leq 1 + (m-1)\beta = 1 + (m-1)\left(1 - \frac{3m\rho}{4m - 1}\right).
\end{align}
Note that the same analysis may apply to $LPT^{\prime\prime}$. Hence, algorithm \Algo has a performance guarantee equal to $\min \{1 + (m-1)(1 - \frac{3m\rho}{4m - 1}); \rho \}$. This expression reaches its largest value when the two terms are equal, namely:
\begin{align}
\label{EqRho}
& 1 + (m-1)\left(1 - \frac{3m\rho}{4m - 1}\right) = \rho
\end{align}
From condition \eqref{EqRho} we derive
\begin{align}
& \rho = \frac{4m -1}{1 + 3m - \frac{1}{m}} = \frac{4}{3} - \frac{7m - 4}{3(3m^2 + m - 1)}\geq \frac{C_{m}^{\Algo}}{C_{m}^*}. \nonumber
\end{align}
\end{proof}

It easy to check that the bound of Proposition \ref{OtherJobs} is strictly inferior to $\frac{4}{3} - \frac{1}{3(m-1)}$ for $m \geq 3$. Thus, we assume in our analysis that any LPT schedule has the last job $n$ as critical job, i.e. $j^\prime = n$. \textit{This assumption is kept in the results of all next propositions \ref{CrucialIneq}--\ref{m2n5}}.

We proceed now with the analysis of instances with $2m+1$ jobs where LPT must couple jobs $1,\dots, m$ respectively with jobs $2m, \dots, m+1$ on the $m$ machines before scheduling job $2m + 1$. Otherwise, Proposition \ref{LPT2m+1} and correspondingly an approximation ratio $\leq \frac{4}{3} - \frac{1}{3(m-1)}$ would hold. Therefore, we will consider the following LPT schedules with pair of jobs on each machine $M_i$ $(i=1,\dots,m)$
\begin{align}
& M_1: p_{1}, p_{2m} \nonumber \\
& M_2: p_{2}, p_{2m-1} \nonumber \\
& \dots \nonumber \\
& M_{m-1}: p_{m - 1},  p_{m + 2} \nonumber \\
& M_m: p_{m},  p_{m + 1} \nonumber 
\end{align}
where job $2m+1$ will be assigned to the machine with the least completion time. We analyze the following two subcases.

\medskip
\paragraph{\textit{Case 1}:  $p_{(2m+1)} \geq p_{1} - p_{m}$}\mbox{} \\
In this case, the last job  $2m+1$ has a processing time greater than (or equal to) the difference $ p_{1} - p_{m}$. Consider $LPT^\prime$ heuristic with $j^\prime = 2m + 1$. The heuristic will assign jobs $2m+1, 1,\dots,m-1$ to machines $M_1, M_2, \dots M_m$ respectively. Then, job $m$ will be loaded on $M_1$ together with job $2m+1$.
Since $p_{(2m+1)} + p_{m} \geq p_{1} $, job $m + 1$ will be processed on the last machine $M_m$ after job $m-1$. Now we have $$p_{(m-1)} + p_{(m+1)} \geq p_{(2m+1)} + p_{m} \geq p_{1} $$ 
since $p_{(m-1)} \geq p_{m}$ and $p_{(m+1)} \geq p_{(2m+1)}$. Hence, job $m+2$ is loaded on machine $M_{(m-1)}$ with job $m-2$. Similarly as before, it follows that $p_{(m-2)} + p_{(m+2)} \geq p_{(2m+1)} + p_{m}$. Consequently, job $m+3$ is processed on $M_{(m-2)}$ after job $m-3$. By applying the same argument, $LPT^\prime$ will assign pair of jobs to each machine until job $2m - 1$ is assigned to $M_2$. 
Eventually, job $2m$ will be assigned to the first machine since it will be the least loaded machine at that point. Summarizing, $LPT^\prime$ will provide the following schedule: 
\begin{align}
& M_1: p_{(2m+1)}, p_m,  p_{2m} \nonumber \\
& M_2: p_{1}, p_{(2m-1)} \nonumber \\
& M_3: p_{2}, p_{(2m-2)} \nonumber \\
& \dots \nonumber \\
& M_{(m-1)}: p_{(m - 2)},  p_{(m + 2)} \nonumber \\
& M_m: p_{(m - 1)},  p_{(m + 1)} \nonumber 
\end{align}
Assume now that the critical machine is $M_1$ with completion time equal to $p_{(2m+1)} + p_m + p_{2m}$. The following proposition holds:
\begin{proposition}
\label{CrucialIneq}
If $p_{(2m+1)} \geq p_{1} - p_{m}$ and $C_{m}^{LPT^\prime} = p_{(2m+1)} + p_m + p_{2m} > C_{m}^*$,  
then $C_{m}^* \geq p_{(m-1)} + p_m$ in any optimal schedule. 
\end{proposition} 
\begin{proof}
We prove the claim by contradiction. We assume that an optimal schedule assigns jobs $1, 2, \dots, m$ to different machines or else $C_{m}^* \geq p_{(m-1)} + p_m$ immediately holds. Correspondingly, since there exists a machine processing three jobs, the optimal makespan can be lower bounded by $p_m + p_{2m} + p_{(2m+1)}$. But as $p_m + p_{2m} + p_{(2m+1)} > C_{m}^*$ holds, a contradiction on the optimality of the schedule is implied. 
\end{proof}
The following proposition also holds.
\begin{proposition}
\label{Slack7/6}
If $p_{(2m+1)} \geq p_{1} - p_{m}$ and $C_{m}^{LPT^\prime} = p_{(2m+1)} + p_m + p_{2m}$, then algorithm \Algo has an approximation ratio not superior to $\frac{7}{6}$.
\end{proposition}
\begin{proof}
We again employ Linear Programming to evaluate the performance of $LPT^\prime$. More precisely, we consider an LP formulation with non-negative variables $p_j$ ($j=1,\dots,n$) denoting the processing times and a positive parameter $OPT > 0$ associated with $C_{m}^*$. 
The corresponding LP model for evaluating the worst case performance of $LPT^\prime$ heuristic is as follows:
{\small
\begin{align}
   \text{maximize}\quad & p_{(2m+1)} + p_m + p_{2m}\label{eq:ObjAlgo}\\
	\text{subject to}\quad
	& p_{(m-1)} + p_m  \leq OPT  \label{eq:FirstMinVal}\\
	& p_{(2m-1)} + p_{2m} + p_{(2m+1)} \leq OPT \label{eq:SecMinVal}\\
	& p_{(2m+1)} - (p_{1} - p_{m}) \geq 0 \label{eq:2.a1} \\
	& p_1 - p_{(m-1)} \geq 0 \label{eq:pjsorted1}\\
	& p_{(m-1)} - p_m \geq 0  \label{eq:pjsorted2}\\
	& p_m - p_{(m+1)} \geq 0 \label{eq:pjsorted3}\\
	& p_{(m+1)} - p_{(2m-1)} \geq 0 \label{eq:pjsorted4}\\
    & p_{(2m-1)} - p_{2m} \geq 0  \label{eq:pjsorted5}\\
     & p_{2m} - p_{(2m + 1)} \geq 0  \label{eq:pjsorted6}\\
	& p_1, p_{(m-1)}, p_m, p_{(m+1)}, p_{(2m-1)}, p_{2m}, p_{(2m+1)}  \geq 0 \label{eq:xjdef} 
\end{align}}
The objective function value (\ref{eq:ObjAlgo}) represents an upper bound on the worst case performance of the algorithm. 
Constraints (\ref{eq:FirstMinVal})--(\ref{eq:SecMinVal}) state that the optimal value $C_{m}^*$ is lower bounded according to Proposition \ref{CrucialIneq} and by the sum of the three jobs with the smallest processing times. 
Constraint (\ref{eq:2.a1}) simply represents the initial assumption $p_{(2m+1)} \geq p_{1} - p_{m}$. Constraints (\ref{eq:pjsorted1})--(\ref{eq:pjsorted6}) state that the considered relevant jobs are sorted by non-increasing processing times while constraints (\ref{eq:xjdef}) indicate that the variables are non-negative. 
We remark that parameter $OPT$ can be arbitrarily set to any value $> 0$. Further valid inequalities (such as $p_{(2m+1)} + p_m + p_{2m} \geq p_{1} + p_{(2m-1)}$ or $OPT \geq \frac{\sum\limits_{j=1}^{n} p_j}{m}$) were omitted as they do not lead to any improvement on the worst case performance ratio. Notice that the number of variables of the reduced LP formulation is constant for any value of $m$.\\
By setting w.l.o.g. $OPT = 1$ and solving model (\ref{eq:ObjAlgo})--(\ref{eq:xjdef}), we get an optimal solution value $z^*$ equal to $1.1666... = \frac{7}{6}$. Correspondingly, the approximation ratio is $\frac{z^*}{OPT} = \frac{7}{6}$.
\end{proof}
Consider now the case where the makespan of $LPT^\prime$ schedule is given by one of the machines $M_2, \dots, M_m$. In such a case, a trivial upper bound on  $LPT^\prime$ makespan is equal to $p_1 + p_{m+1}$. 
We state the following proposition.
\begin{proposition}
\label{First2m-1}
If $p_{(2m+1)} \geq p_{1} - p_{m}$ and the makespan of $LPT^\prime$ is not on $M_1$, then coupling $LPT$ with $LPT^\prime$ gives a performance guarantee not superior to $\frac{15}{13}$ for $m=3$ and $\frac{4}{3} - \frac{1}{2m-1}$ for $m \geq 4$.
\end{proposition}
\begin{proof}
We again consider an LP formulation with non-negative variables $p_j$ ($j=1,\dots,n$), a positive parameter $OPT > 0$ and two non-negative auxiliary variables $\alpha, y$. We can evaluate the worst case performance of $LPT + LPT^\prime$ by the following LP model
{\small
\begin{align}
   \text{maximize}\quad & y \label{eq:ObjCase3bis}\\
	\text{subject to}\quad
	& \sum\limits_{j=1}^{2m + 1} p_j \leq mOPT  \label{eq:FirstOPT}\\
	& p_{(2m-1)} + p_{2m} + p_{(2m+1)} \leq OPT \label{eq:SecOPT}\\
	& p_{(j + 1)} - p_j \leq 0 \quad j=1, \dots, 2m; \label{eq:pjsort}\\
	& p_{j} + p_{(2m - j + 1)} - \alpha \geq 0 \quad j=1, \dots, m; \label{eq:LPT2m1} \\
	& p_{(2m+1)} + \alpha - y \geq 0 \label{eq:LPTobj}\\
	& p_{(2m+1)} - (p_1 - p_m) \geq 0 \label{eq:CaseCond}\\
     & p_{1} + p_{(m + 1)} - y \geq 0  \label{eq:UBLPTprime}\\
	& p_j \geq 0 \quad j=1, \dots, 2m + 1; \label{eq:xj} \\
	& \alpha, y \geq 0 \label{eq:alpha_y} 
\end{align}}
where $y$ represents the solution value reached by $LPT + LPT^\prime$ and $\alpha$ is the starting time of job $2m + 1$ in LPT.
Constraints \eqref{eq:LPT2m1} indicate that $\alpha$ corresponds to the smallest load on a machine after processing jobs $1,\dots,2m$. The solution value of LPT is therefore the sum $\alpha + p_{(2m + 1)}$. The objective function \eqref{eq:ObjCase3bis} provides an upper bound on the makespan of $LPT + LPT^\prime$ since it maximizes the minimum between $\alpha + p_{(2m + 1)}$ and the makespan reached by $LPT^\prime$ 
through variable $y$ and related constraints \eqref{eq:LPTobj} and \eqref{eq:UBLPTprime}. Constraints \eqref{eq:FirstOPT} and \eqref{eq:SecOPT} state that $C_{m}^*$ is lower bounded by $\frac{\sum\limits_{j=1}^{n} p_j}{m}$ and by $p_{(2m - 1)} + p_{2m} + p_{(2m + 1)}$. Constraints \eqref{eq:pjsort} indicate that jobs are sorted by non-increasing processing time while constraint \eqref{eq:CaseCond} represents condition $p_{(2m+1)} \geq p_{1} - p_{m}$. Finally, constraints \eqref{eq:xj} and \eqref{eq:alpha_y} indicate that all variables are non negative.\\
\smallskip
By setting w.l.o.g. $OPT = 1$, a feasible solution of model \eqref{eq:ObjCase3bis}--\eqref{eq:alpha_y} for any value of $m$ is:
\begin{align}
 & y = \frac{8m - 7}{3(2m - 1)}; ~ \alpha = \frac{2(m - 1)}{2m - 1}; \nonumber \\
 & p_1 = \frac{5m - 4}{3(2m - 1)}; ~ p_2 = p_3 = \dots = p_{(m-1)} = \frac{4m - 5}{3(2m - 1)}; \nonumber \\
 & p_m = p_{m + 1} = \frac{m - 1}{2m - 1}; ~ p_{m + 2} = p_{m + 3} = \dots = p_{2m + 1} = \frac{1}{3}.  \nonumber 
\end{align}
We can show by strong duality that this solution is optimal for any $m \geq 4$. The dual model with variables $\lambda_i$ $(i= 1,\dots, 3m + 5)$ associated with constraints \eqref{eq:FirstOPT}--\eqref{eq:UBLPTprime} is as follows:
{\small
\begin{align}
   \text{minimize}\quad & m\lambda_1  + \lambda_2 \label{eq:ObjDual2}\\
	\text{subject to}\quad
	& \lambda_1 - \lambda_3 + \lambda_{(2m + 3)}  - \lambda_{(3m + 4)} + \lambda_{(3m + 5)} \geq 0  \label{eq:p1} \\
	& \lambda_1 + \lambda_{(1 + j)} - \lambda_{(2 + j)} + \lambda_{(2m + 2 + j)} \geq 0  \quad j=2, \dots, m - 1 \label{eq:p2-pm} \\
	& \lambda_1 + \lambda_{(m + 1)} - \lambda_{(m +2)} + \lambda_{(3m + 2)} +\lambda_{(3m + 4)}  \geq 0 \label{eq:pm} \\
	& \lambda_1 + \lambda_{(m + 2)} - \lambda_{(m + 3)} + \lambda_{(3m + 2)} +\lambda_{(3m + 5)}  \geq 0 \label{eq:pm+1} \\
& \lambda_1 + \lambda_{(1 + j)} - \lambda_{(2 + j)} + \lambda_{(4m + 3 - j)} \geq 0  \quad j=m+2, \dots, 2m - 2 \label{eq:pm+2-p2m-2} \\
& \lambda_1 + \lambda_{2} + \lambda_{2m} - \lambda_{(2m + 1)} + \lambda_{(2m + 4)} \geq 0  \label{eq:p2m-1} \\
& \lambda_1 + \lambda_{2} + \lambda_{(2m+1)} - \lambda_{(2m+2)} + \lambda_{(2m + 3)} \geq 0  \label{eq:p2m} \\
& \lambda_1 + \lambda_{2} + \lambda_{(2m+2)} + \lambda_{(3m + 3)} + \lambda_{(3m + 4)}  \geq 0  \label{eq:p2m+1} \\
& -\sum\limits_{j=(2m + 3)}^{(3m + 2)}\lambda_j + \lambda_{(3m + 3)} \geq 0  \label{eq:alpha} \\
& - \lambda_{(3m + 3)} - \lambda_{(3m + 5)}  \geq 1  \label{eq:y} \\
	& \lambda_1, \lambda_2, \dots, \lambda_{(2m + 2)} \geq 0  \label{eq:TheDualVar3.1} \\
	& \lambda_{(2m + 3)} , \lambda_{(2m + 4)}, \dots \lambda_{(3m + 5)} \leq 0 \label{eq:TheDualVar3.2}
\end{align}}
Constraints \eqref{eq:p1}--\eqref{eq:y} correspond to primal variables $p_j, \alpha, y$ respectively.
A feasible solution of model \eqref{eq:ObjDual2}--\eqref{eq:TheDualVar3.2} for $m \geq 4$ is:
\begin{align}
 & \lambda_1 = \frac{2}{2m - 1}; ~  \lambda_2 = \frac{2m - 7}{3(2m - 1)}; ~ \lambda_3 =  \lambda_4 = \dots = \lambda_{(m+1)} = 0; \nonumber \\
     & \lambda_{(m+2)} = \frac{1}{2m - 1}; ~ \lambda_{(m+3)} = \lambda_{(m+4)} = \dots = \lambda_{2m} = 0;  \lambda_{(2m + 1)} = \frac{2m - 7}{3(2m - 1)}; \nonumber \\
 & \lambda_{(2m + 2)} = \frac{4(m - 2)}{3(2m - 1)}; ~ \lambda_{(2m+3)} =  0; ~ \lambda_{(2m+4)} = \lambda_{(2m+5)} = \dots = \lambda_{(3m + 1)} = \frac{-2}{2m - 1};  \nonumber \\
  & \lambda_{(3m + 2)} = \frac{-1}{2m - 1}; ~ \lambda_{(3m + 3)} = \frac{3 - 2m}{2m - 1}; ~ \lambda_{(3m + 4)} = 0; ~ \lambda_{(3m + 5)} =  \frac{-2}{2m - 1}.  \nonumber
\end{align}
The corresponding solution value is $m\lambda_1  + \lambda_2 = \frac{8m - 7}{3(2m - 1)} = y$. Hence,  for $m \geq 4$ we have:
$$\frac{\min\{C_{m}^{LPT},C_{m}^{LPT^\prime}\}}{C_{m}^*} \leq \frac{y}{OPT} = \frac{8m - 7}{3(2m - 1)} = \frac{4(2m - 1) - 3}{3(2m - 1)} = \frac{4}{3} - \frac{1}{2m-1}$$
For $m=3$, an optimal solution of model \eqref{eq:ObjCase3bis}--\eqref{eq:alpha_y} has value $y = 1.15385...=\frac{15}{13}$. 
The corresponding approximation ratio is then $\frac{y}{OPT} = \frac{15}{13}$.
\end{proof}
\medskip
\paragraph{\textit{Case 2}: $p_{(2m+1)} < p_{1} - p_{m}$}\mbox{}\\
The processing time of the last job is smaller than the difference $p_{1} - p_{m}$. The following proposition holds.
\begin{proposition}
\label{(2m+1)<1-m}
If $p_{(2m+1)} < p_{1} - p_{m}$, the solution given by $LPT$ has a performance guarantee not superior to $\frac{15}{13}$ for $m=3$ and $\frac{4}{3} - \frac{1}{2m-1}$ for $m \geq 4$.
\end{proposition}
\begin{proof}
We consider the worst case LPT performance and notice that it can be evaluated through model \eqref{eq:ObjCase3bis}--\eqref{eq:alpha_y} by simply reversing the inequality sign in constraint \eqref{eq:CaseCond} and disregarding constraint \eqref{eq:UBLPTprime}. Correspondingly, dual model \eqref{eq:ObjDual2}--\eqref{eq:TheDualVar3.2} is still implied with the differences that variable $\lambda_{(3m + 5)}$ is discarded and that we have $\lambda_{(3m + 4)} \geq 0$. The primal solutions turn out to be the same solutions stated in Proposition \ref{First2m-1}. Likewise, dual solutions slightly modify in the following variables entries which do not contribute to the objective function: $\lambda_{(3m + 2)} =  \frac{-3}{2m - 1}; \lambda_{(3m + 3)} = -1;  \lambda_{(3m + 4)} =  \frac{2}{2m - 1}$. Therefore, the approximation ratios stated in Proposition \ref{First2m-1} hold. 
\end{proof}
We can now state the following theorem for \Algo.
\begin{theorem}
\label{OurTheorem}
Algorithm \Algo runs in $\mathcal{O}(n \log n)$ time and has an approximation ratio not superior to $\frac{4}{3} - \frac{1}{3(m-1)}$ for $m \geq 3$.
\end{theorem}
\begin{proof}
Putting together the results of propositions \ref{Slack7/6}, \ref{First2m-1}, \ref{(2m+1)<1-m}, it immediately follows that \Algo has an approximation ratio not superior to $\frac{4}{3} - \frac{1}{3(m-1)}$ for $m \geq 3$.  
Besides, it is well known that the running time of the LPT heuristic is in $\mathcal{O}(n \log n)$: sorting the jobs by processing time has complexity $\mathcal{O}(n \log n)$ while an efficient implementation of the underlying LS procedure may employ a Fibonacci's heap for extracting the machine with smallest load at each iteration with overall complexity $\mathcal{O}(n \log m)$. Since the proposed algorithm requires to run first LPT (to compute $z_1$) and then LS heuristic twice (to compute $z_2$ and $z_3$) after the job sorting, the resulting time complexity is $\mathcal{O}(n \log n)$. 
\end{proof}
We remark that in the problem variant with two machines $(m=2)$, the approximation ratio of $\frac{4}{3} - \frac{1}{3(m-1)} = 1$ cannot be reached by any polynomial time algorithm unless $\mathcal{P} = \mathcal{NP}$, since $P_2 || C_{max}$ is well known to be NP--Hard. Still, the following analysis shows that \Algo has a better performance guarantee than the bound of $\frac{7}{6}$ implied for LPT when $m=2$.
\subsubsection{\Algo performance analysis for $P_2 || C_{max}$}
We show here that algorithm \Algo has a performance guarantee of $\frac{9}{8}$. We proceed by assuming that the critical job in the LPT solution is the last one $(j^\prime = n)$. Otherwise, we would get an approximation ratio of $\frac{14}{13} < \frac{9}{8}$ according to Proposition \ref{OtherJobs}. Given Lemma \ref{BoChenImplied} and bound \eqref{GenBound}, an approximation bound worse than $\frac{9}{8}$ could be reached only in instances with $n \leq 6$ and where LPT schedules no more than three jobs on each machine.
\\
For $n \leq 4$, $LPT$ will output an optimal solution according to Proposition \ref{LPT2m}. For $n = 6$ , we can evaluate the worst-case LPT performance by model \eqref{eq:1_3m}--\eqref{eq:TheVar_3m} (in the appendix) 
since $n=3m$. The corresponding optimal objective value for $m=2$ is $opt = 0.8888...=\frac{8}{9}$ which gives an approximation ratio equal to $\frac{1}{opt} = \frac{9}{8}$. For $n = 5$, we have the following proposition.
\begin{proposition}
\label{m2n5}
In $P_2 || C_{max}$ instances with $n = 5$, \Algo provides an optimal solution. 
\end{proposition}
\begin{proof}
According to Proposition \ref{LPT2m+1}, LPT could not give the optimal makespan only if it assigns jobs 1 and 4 to $M_1$ and jobs 2 and 3 to $M_2$ before assigning the critical job 5. Considering this LPT schedule, we analyze all possible optimal solution configurations and show that \Algo always identifies them. We have the following two cases:
\begin{itemize}
\item \textit{Case 1: jobs 1 and 2 are on the same machine in an optimal solution.}\\
There exists an optimal solution which assigns jobs 3, 4, 5 to $M_1$ and jobs 1, 2 to $M_2$. In fact, any other schedule cannot provide a smaller makespan. The same solution is also provided by $LPT^{\prime\prime}$.
\item \textit{Case 2: jobs 1 and 2 are on different machines in an optimal solution}.\\
If there exists an optimal solution with job 1 on $M_1$ and jobs 2 and 3 on $M_2$, such a solution coincides with the LPT schedule. If instead an optimal solution assigns jobs 1 and 3 to $M_1$, then it must assign jobs 2, 4 and 5 to $M_2$. If $p_3 = p_4$, clearly LPT also gives the optimal makespan. If $p_3 > p_4$, inequality $p_1 \leq p_2 + p_5$ must hold or else the optimal solution would be improved by exchanging job $3$ with job $4$ on the machines giving a contradiction. But then, inequality $p_1 \leq p_2 + p_5$ implies that $LPT^\prime$ solution is also optimal since it assigns jobs 5, 2 and 4 to $M_1$ and jobs 1 and 3 to $M_2$.
\end{itemize}
\end{proof}

\section{From approximation to heuristics: a new LPT--based approach} 
It can be noted that in our theoretical results $LPT^\prime$ was necessary to improve Graham's bound for $m\geq 3$ while $LPT^{\prime\prime}$ was necessary for $m = 2$ only.
Remarkably, for $m \geq 3$ the relevant subcase was the one with $2m + 1$ jobs, $p_{2m+1} \geq p_1-p_m$ and $LPT^\prime$ required to schedule job $2m+1$ first and then to apply List Scheduling first to the sorted job subset $\{1,...,m\}$ and then to the sorted job subset $\{m+1,...,2m\}$. 
This suggests a general greedy approach that considers not only the ordering of the jobs but also the differences in processing time within job subsets of size $m$. We propose a constructive procedure that splits the sorted jobset in tuples of $m$ consecutive jobs ($1,\dots,m; m+1,\dots,2m;$ etc.) and sorts the tuples in non-increasing order of the difference between the largest job and the smallest job in the tuple. Then, List Scheduling is applied to the set of sorted tuples. 
We denote this approach as \Slack.
\begin{algorithm}[H] 
\begin{algorithmic}[1]
\STATEx \textbf{Input:} $P_m || C_{max}$ instance with $m$ machines and $n$ jobs with processing times $p_j$ $(j = 1,\dots,n)$.
\STATE Sort jobs by non-increasing $p_j$.
\STATE Consider $\Bigl\lceil \frac{n}{m} \Bigr\rceil $ tuples with size $m$ given by jobs $1,\dots,m;m+1,\dots, 2m$, etc.. If $n$ is not multiple of $m$, add dummy jobs with null processing time in the last tuple.
\STATE For each tuple,  compute the associated slack, namely $p_1 - p_m, p_{(m+1)} - p_{2m}, \dots, p_{(n-m+1)} - p_{n}$.
\STATE Sort tuples by non-increasing slack and then fill a list with consecutive jobs in the sorted tuples. 
\STATE Apply List Scheduling to this job ordering and return the solution.
\end{algorithmic}
\caption{\textbf{\Slack heuristic}}
\end{algorithm}	
In terms of computational complexity, since construction and sorting of the tuples can be performed in $\mathcal{O}(\Bigl\lceil \frac{n}{m} \Bigr\rceil  \log \Bigl\lceil \frac{n}{m} \Bigr\rceil )$, the running time of \Slack is $\mathcal{O}(n \log n)$ due to the initial sorting of the jobs. \\
We compare \Slack to LPT on benchmark literature instances provided in \cite{Iori2008}.
All tests have been conducted on an Intel i5 CPU @ 2.30 GHz with 8 GB of RAM. Both algorithms have been implemented in C++.\\
In \cite{Iori2008} two classical classes of instances from literature are considered: \textit{uniform instances} proposed in \cite{FraGeLaMu94}; \textit{non-uniform instances} proposed in \cite{FraNeSc04}. In \textit{uniform instances}  the processing times are integer uniformly distributed in the range $[a, b]$. In \textit{non-uniform instances}, $98\%$ of the processing times are integer uniformly distributed in $[0.9(b-a), b]$ while the remaining ones are uniformly distributed in $[a, 0.2(b-a)]$. For both classes, we have $a = 1; b = 100, 1000, 10000$. For each class, the following values were considered for the number of machines and jobs: $m = 5, 10, 25$ and $n = 10, 50, 100, 500, 1000$.
For each pair $(m,n)$ with $m < n$, 10 instances were generated for a total of 780 instances.\\
We compared the performance of the algorithms by counting how many times \Slack improves (resp. worsens) the solution value provided by LPT 
or yields the same makespan. 
The results are reported in Table \ref{tab:LitInsLPT} where instances are aggregated by processing time range and number of machines as in \cite{Iori2008}. Running times of the heuristics are not reported since are negligible.
\begin{table}[H]
	\centering
	\small
	\scalebox{0.85}{
\begin{tabular}{|ll|c|*{4}{c|}*{2}{c|}*{2}{c|}*{2}{c|} }
  \hline 
&  & NON UNIFORM& \multicolumn{2}{|c|}{\Slack } & \multicolumn{2}{|c|}{}& \multicolumn{2}{|c|}{LPT}\\ 
 &  & Instances & \multicolumn{2}{|c|}{wins} & \multicolumn{2}{|c|}{draws}& \multicolumn{2}{|c|}{wins}\\	\hline
 $[a, b]$ &	$m$  &  \# &  \# & (\%) &  \# & (\%) &  \# & (\%) \\	\hline	
        & 5  & 50 & 31 & (62.0)  & 16 & (32.0) & 3 & (6.0) \\
1-100   & 10 & 40 & 32 & (80.0)  & 8  & (20.0) & 0 & (0.0) \\
        & 25 & 40 & 23 & (57.5)  & 17 & (42.5) & 0 & (0.0) \\ \hline
        & 5  & 50 & 39 & (78.0)  & 10 & (20.0) & 1 & (2.0) \\
1-1000  & 10 & 40 & 40 & (100.0) & 0  & (0.0)  & 0 & (0.0) \\
        & 25 & 40 &27 & (67.5)  & 12 & (30.0) & 1 & (2.5) \\ \hline
        & 5  & 50 &39 & (78.0)  & 10 & (20.0) & 1 & (2.0) \\
1-10000 & 10 & 40 & 40 & (100.0) & 0  & (0.0)  & 0 & (0.0) \\
        & 25 & 40 & 28 & (70.0)  & 10 & (25.0) & 2 & (5.0) \\ \hline \hline
 &  & UNIFORM & \multicolumn{2}{|c|}{\Slack } & \multicolumn{2}{|c|}{}& \multicolumn{2}{|c|}{LPT}\\ 
 &  & Instances & \multicolumn{2}{|c|}{wins} & \multicolumn{2}{|c|}{draws}& \multicolumn{2}{|c|}{wins}\\	\hline
 $[a, b]$ &	$m$  & \# & \# & (\%) &  \# & (\%) &  \# & (\%) \\	\hline	
        & 5  & 50 & 12 & (24.0) & 37 & (74.0) & 1 & (2.0)  \\
1-100   & 10 & 40 & 14 & (35.0) & 20 & (50.0) & 6 & (15.0) \\
        & 25 & 40 & 10 & (25.0) & 29 & (72.5) & 1 & (2.5)  \\ \hline
        & 5  & 50 & 32 & (64.0) & 15 & (30.0) & 3 & (6.0)  \\
1-1000  & 10 & 40 & 27 & (67.5) & 5  & (12.5) & 8 & (20.0) \\
        & 25 & 40 & 24 & (60.0) & 12 & (30.0) & 4 & (10.0) \\ \hline
        & 5  & 50 & 36 & (72.0) & 12 & (24.0) & 2 & (4.0)  \\
1-10000 & 10 & 40 &  37 & (92.5) & 0  & (0.0)  & 3 & (7.5)  \\
        & 25 & 40 & 22 & (55.0) & 11 & (27.5) & 7 & (17.5)	\\	 \hline	
\end{tabular}}
		\caption{\Slack vs LPT: performance comparison on $P_m || C_{max}$ \text{instances} from \cite{Iori2008}.}
		\label{tab:LitInsLPT}
\end{table}
\Slack algorithm strongly outperforms LPT rule in each instance category with the most impressive performance difference on non uniform instances. Overall, on 780 benchmark literature instances,  \Slack wins 513 times (65.8\% of the cases) against LPT, ties 224 times (28.7\%) and loses 43 times (5.5\%) only. Given these results, \Slack heuristic can be regarded as a valuable alternative to the popular LPT rule. \\
To get a broader picture, we also compared \Slack against COMBINE, an algorithm proposed in \cite{LeeMa88}, that couples LPT with the MULTIFIT heuristic introduced in \cite{CofGaJo78}. With an increase of the computational effort, COMBINE generally exhibits better performances than LPT (see, e.g., \cite{LeeMa88}). The results of the performance comparison between \Slack and COMBINE on the considered benchmark instances are reported in Table \ref{tab:LitInsCombine}.
\begin{table}[H]
	\centering
	\small
	\scalebox{0.85}{
\begin{tabular}{|ll|c|*{4}{c|}*{2}{c|}*{2}{c|}*{2}{c|} }
  \hline 
&  & NON UNIFORM& \multicolumn{2}{|c|}{\Slack } & \multicolumn{2}{|c|}{}& \multicolumn{2}{|c|}{COMBINE}\\ 
 &  & Instances & \multicolumn{2}{|c|}{wins} & \multicolumn{2}{|c|}{draws}& \multicolumn{2}{|c|}{wins}\\	\hline
 $[a, b]$ &	$m$  &  \# &  \# & (\%) &  \# & (\%) &  \# & (\%) \\	\hline	
         & 5  & 50 & 30 & (60.0) & 16 & (32.0) & 4 & (8.0)  \\
1-100   & 10 & 40 & 29 & (72.5) & 8  & (20.0) & 3 & (7.5)  \\
        & 25 & 40 & 23 & (57.5) & 17 & (42.5) & 0 & (0.0)  \\  \hline
        & 5  & 50& 39 & (78.0) & 10 & (20.0) & 1 & (2.0)  \\
1-1000  & 10 & 40& 33 & (82.5) & 0  & (0.0)  & 7 & (17.5) \\
        & 25 & 40& 27 & (67.5) & 12 & (30.0) & 1 & (2.5)  \\  \hline
        & 5  & 50& 39 & (78.0) & 10 & (20.0) & 1 & (2.0)  \\
1-10000 & 10 & 40& 32 & (80.0) & 0  & (0.0)  & 8 & (20.0) \\
        & 25 & 40& 28 & (70.0) & 10 & (25.0) & 2 & (5.0) \\  \hline  \hline
 &  & UNIFORM & \multicolumn{2}{|c|}{\Slack } & \multicolumn{2}{|c|}{}& \multicolumn{2}{|c|}{COMBINE}\\ 
 &  &  Instances & \multicolumn{2}{|c|}{wins} & \multicolumn{2}{|c|}{draws}& \multicolumn{2}{|c|}{wins}\\	\hline
 $[a, b]$ &	$m$  & \# & \# & (\%) &  \# & (\%) &  \# & (\%) \\	\hline	
        & 5  & 50 & 10 & (20.0) & 38 & (76.0) & 2  & (4.0)  \\
1-100   & 10 & 40 & 9  & (22.5) & 21 & (52.5) & 10 & (25.0) \\
        & 25 & 40 & 4  & (10.0) & 23 & (57.5) & 13 & (32.5) \\  \hline
        & 5  & 50 & 31 & (62.0) & 15 & (30.0) & 4  & (8.0)  \\
1-1000  & 10 & 40 & 21 & (52.5) & 5  & (12.5) & 14 & (35.0) \\
        & 25 & 40 & 17 & (42.5) & 6  & (15.0) & 17 & (42.5) \\  \hline
        & 5  & 50 & 36 & (72.0) & 11 & (22.0) & 3  & (6.0)  \\
1-10000 & 10 & 40 & 30 & (75.0) & 0  & (0.0)  & 10 & (25.0) \\
        & 25 & 40 & 15 & (37.5) & 6  & (15.0) & 19 & (47.5) \\  \hline
\end{tabular}}
		\caption{\Slack vs COMBINE: performance comparison on $P_m || C_{max}$ \text{instances} from \cite{Iori2008}.}
		\label{tab:LitInsCombine}
\end{table}
The results in Table \ref{tab:LitInsCombine} provide additional evidence on the effectiveness of \Slack, which outperforms COMBINE on non uniform instances and favorably compares to the competing algorithm on uniform instances. Out of 780 instances, \Slack wins against COMBINE 453 times (58.1\% of the cases), ties 208 times (26.7\%) and loses 119 times (15.3\%).

\section{Concluding remarks}
We provided new insights on the well-known LPT rule for $P_m || C_{max}$ problem and proposed a slight algorithmic revisiting which improves previously published approximation ratios for LPT. As second major contribution, from our approximation results we came up with a simple heuristic which strongly outperforms LPT on a large set of benchmark literature instances. \\    
In our analysis of $P_m || C_{max}$, we deployed a novel approach which relies on Linear Programming. The proposed LP reasoning could be considered a valid alternative to techniques based on analytical derivation and may as well find application in other combinatorial optimization problems. For example, an attempt in this direction has been recently proposed in \cite{DEPFSC17} for a multiperiod variant of the knapsack problem.\\
We remark that in this work we did not derive tight approximation bounds for \Algo algorithm. We reasonably expect that improved bounds can be stated and leave this issue to future research. Nonetheless, we found out $P_m || C_{max}$ instances for $m \geq 3$ which provide a lower bound on the worst case performance ratio of \Algo equal to $\frac{4}{3} - \frac{7}{3(3m + 1)}$. These instances have $2m+2$ jobs with processing times:
\begin{align}
& p_j = 2m - \left\lfloor \frac{j + 1}{2} \right\rfloor, \quad 1 \leq j \leq 2m - 2; \nonumber \\
& p_j = m, \quad 2m - 1 \leq j \leq 2m + 2 \nonumber
\end{align}
It is easy to check that $C_{m}^{ \Algo} = 4m - 1$ and $C_{m}^{*} = 3m + 1$ and that such values give the above performance ratio.

\section{Appendix: Proof of Proposition \ref{mthfo}}
\label{Apx1} 
\begin{proof}
The relevant cases involve instances with $2m + 2 \leq n \leq 3m$, where LPT schedules the critical job in third position on a machine. Also, notice that each non critical machine must process at most three jobs to contradict the claim, otherwise the results of Lemma \ref{BoChenImplied} holds for $k \geq 4$. In addition, an optimal solution must assign  at most three jobs to each machine. Otherwise, we would have $ C_{m}^* \geq \sum\limits_{j=n-3}^{n} p_j \implies p_n \leq \frac{C_{m}^*}{4}$  which induces a $\left(\frac{5}{4} - \frac{1}{4m}\right)$--performance guarantee according to expression (\ref{GenRel}) (with $j^\prime = n$).\\
Considering the above requirements for both LPT schedule and the optimal solution, we introduce different LP formulations which consider appropriate bounds on the completion times of the machines as well as on the optimal makespan according to the number of jobs and machines involved. We analyze two macro-cases defined by the two different LP modelings.  
\begin{itemize}
\item Case a): $n = 3m ~(m=3 ~ n=9; m=4 ~ n=12);$ 
\item Case b): $2m+2 \leq n \leq 3m - 1  ~(m=3 ~ n=8; m=4 ~ n=10,11).$
\end{itemize}
\subsection{Case a)}
Since $n=3m$, an optimal solution must process exactly three jobs on each machine  to contradict the claim. This implies a lower bound on the optimum equal to $p_1 + p_{(3m-1)} + p_{3m}$ as well as that condition $p_1 \leq 2p_{3m}$ must hold (otherwise $p_{n} \leq \frac{C_{m}^*}{4}$). Given the last condition, LPT couples jobs $1,2,\dots,m$ respectively with jobs $2m,(2m-1),\dots,(m+1)$ on the machines. A valid upper bound on $C_{m}^{LPT}$ is hence given by $p_1 + p_{(m+1)} + p_{3m}$. In order to evaluate the worst case LPT performance, we introduce a simple LP formulation with non-negative variables $p_j$ ($j=1,\dots,n$) related to job processing times and variable $opt$ which again represents $C_{m}^*$. As in model \eqref{eq:ObjOpt}--\eqref{eq:TheVar}, we minimize the value of $C_{m}^*$ after setting w.l.o.g. $C_{m}^{LPT}=1$. The following LP model is implied: 
{\small
\begin{align}
   \text{minimize}\quad & opt \label{eq:1_3m}\\
	\text{subject to}\quad
	& \sum\limits_{j=1}^{3m} p_j \leq m\cdot opt  \label{eq:2_3m}\\
	& p_1 + p_{(3m -1)} + p_{3m} - opt \leq 0  \label{eq:3_3m}\\
	& p_1 - 2p_{3m} \leq 0  \label{eq:4_3m}\\
    & p_1 + p_{(m + 1)} + p_{3m} \geq 1 \label{eq:heur_3m} \\
    & p_{(j + 1)} - p_j \leq 0 \quad j=1, \dots, 3m - 1; \label{eq:pjsort_3m}\\
    & p_j \geq 0 \quad j=1, \dots, 3m; \label{eq:TheVar_3m}
\end{align}}
Constraints \eqref{eq:2_3m}--\eqref{eq:TheVar_3m} represent the above specified conditions together with the sorting of the jobs by decreasing processing time (constraint \eqref{eq:pjsort_3m}).
Solving model (\ref{eq:1_3m})--(\ref{eq:TheVar_3m}) by means of an LP solver (e.g. CPLEX) provides the optimal solution value of $opt$ for any value of $m$. More precisely, we have $opt=0.8571\dots=\frac{6}{7}$ for $m=3$ and $opt=0.8421\dots=\frac{16}{19}$ for $m=4$. The claim is showed by the corresponding upper bounds $\frac{1}{opt}$ on the ratio $\frac{C_{m}^{LPT}}{C_{m}^*}$ which are equal to $\frac{7}{6}$ for $m=3$ and to $\frac{19}{16} (< \frac{11}{9})$ for $m=4$.
\subsection{Case b)}
As said above, LPT assigns three jobs to the critical machine. The possible values of $n$ and $m$ also imply that there is a non critical machine in the LPT solution which schedules three jobs. As in \textit{Case a)}, we introduce an LP formulation which reconsiders model \eqref{eq:ObjOpt}--\eqref{eq:TheVar}. For the target non critical machine loading three jobs, we distinguish the processing time of the last job assigned to the machine, represented by a non negative variable $p^{\prime}$, from the contribution of the other two jobs, represented by non negative variable $t^{\prime}$.  Variables $t_{c},t^{\prime\prime}, opt$ have the same meaning as in model \eqref{eq:ObjOpt}--\eqref{eq:TheVar}. \\
First notice that $p^{\prime} \geq p_{(n-1)}$ holds as the critical job is $n$ and that neither job $n-1$ nor job $n$ can contribute to the value of $t^{\prime}$. Also, in any LPT schedule both $t_{c}$ and $t^{\prime}$ are given by the sum of two jobs where the first job is between job 1 and job $m$. The following relations straightforwardly hold: 
$$t_{c} \leq p_1 + p_{(m+1)}; ~ t^{\prime} \geq p_m + p_{(n-2)}$$ 
Considering these conditions and the reasoning applied to model \eqref{eq:ObjOpt}--\eqref{eq:TheVar}, we get the following LP model: 
{\small
\begin{align}
   \text{minimize}\quad & opt \label{eq:ObjOptB}\\
	\text{subject to}\quad
	& \sum\limits_{j=1}^{n} p_j - m\cdot opt \leq 0  \label{eq:BoundOptB}\\
     & (t_{c}  + p_n) + (t^{\prime} + p^{\prime}) +  t^{\prime\prime}  - \sum\limits_{j=1}^{n} p_j  = 0 \label{eq:sumpB}\\	
	& t_{c} - (t^{\prime} + p^{\prime}) \leq 0 \label{eq:L1t1B}\\
    & (m-2)t_{c}  - t^{\prime\prime}  \leq 0  \label{eq:avL3B}\\
    & t_{c} + p_n = 1 \label{eq:heurB} \\
    & p_{(n-1)} - p^{\prime}  \leq 0 \label{eq:primeB} \\
    & p_m + p_{(n-2)} - t^{\prime}  \leq 0 \label{eq:2t1B}\\
	& t_{c} - (p_1 + p_{(m+1)}) \leq 0 \label{eq:2L1B}\\
	& p_{(j + 1)} - p_j \leq 0 \quad j=1, \dots, n - 1; \label{eq:pjsort_B}\\
    & p_j \geq 0 \quad j=1, \dots, n; \label{eq:TheVar_B}\\
	& t_{c}, t^{\prime}, t^{\prime\prime}, p^{\prime}, opt \geq 0  \label{eq:TheVarB}
\end{align}}
Model \eqref{eq:ObjOptB}--\eqref{eq:TheVarB} constitutes a backbone LP formulation for all subcases analyzed in the following.
\subsubsection{$P_m || C_{max}$ instances with $m=4$, $n=11$}
To contradict the claim, an optimal solution must assign three jobs to three machines and two jobs to one machine. Assume first that the optimal solution processes the first three jobs on two machines. Since the optimal solution must schedule at least five jobs on two machines, a valid lower bound on $C_{m}^*$ is equal to one half of the sum $(p_1 + p_2 + p_3 + p_{10} + p_{11})$. Thus, we can add constraint 
$$ opt \geq \frac{p_1 + p_2 + p_3 + p_{10} + p_{11}}{2}$$
to model \eqref{eq:ObjOptB}--\eqref{eq:TheVarB}. The corresponding LP optimal solution gives an upper bound on ratio $\frac{C_{m}^{LPT}}{C_{m}^*}$ equal to $\frac{1}{opt} = \frac{11}{9}$.\\
Likewise, should the optimal solution schedule the first three jobs on three different machines, a valid lower bound on $C_{m}^*$ would correspond to the following constraint
$$ opt \geq \frac{p_1 + p_2 + p_3 + p_7 + p_8 + p_9 + p_{10} + p_{11}}{3}$$
since the optimal solution has to process at least eight jobs on three machines. If we add the last constraint to model \eqref{eq:ObjOptB}--\eqref{eq:TheVarB}, the LP optimal solution yields again a value of $\frac{1}{opt}$ equal to $\frac{11}{9}$.

\subsubsection{$P_m || C_{max}$ instances with $m=4$, $n=10$}
An optimal solution either assigns three jobs to three machines and one job to the other machine, or three jobs to two machines and two jobs to the others. We again distinguish whether the optimal solution processes the first three jobs either on two or three machines. 
In the first case, since two machines must process at least four jobs in an optimal solution, a valid lower bound on $C_{m}^*$ is equal to one half of the sum $(p_1 + p_2 + p_3 + p_{10})$. 
The optimal objective value of model \eqref{eq:ObjOptB}--\eqref{eq:TheVarB} after adding constraint 
$$ opt \geq \frac{p_1 + p_2 + p_3 + p_{10}}{2}$$ provides an approximation ratio equal to $\frac{1}{opt} = \frac{11}{9}$.
In the second case, as three machines must necessarily process at least seven jobs, we can add to the LP model the following constraint which represents a valid lower bound on $C_{m}^*$: 
$$ opt \geq \frac{p_1 + p_2 + p_3 + p_7 + p_8 + p_9 + p_{10}}{3}$$
From the corresponding LP optimal solution we get again $\frac{1}{opt}=\frac{11}{9}$.
\subsubsection{$P_m || C_{max}$ instances with $m=3$, $n=8$}
We have to address the case where an optimal solution assigns two jobs to a machine and three jobs to the other two machines. This implies relations $C_{m}^* \geq p_1 + p_8$ and $C_{m}^* \geq \frac{p_1 + p_2 + p_6 + p_7 + p_8}{2}$ and the corresponding constraints to be added to model \eqref{eq:ObjOptB}--\eqref{eq:TheVarB}:  $$ opt \geq p_1 + p_8; ~ opt \geq \frac{p_1 + p_2 + p_6 + p_7 + p_8}{2}$$
For the solution provided by LPT, we now analyze all possible assignments of the jobs which can give $t^{\prime}$. LPT separately schedules the first three jobs and job 4 with job 3 on the third machine. It immediately follows that if the target non critical machine is the first one, then it must necessarily process job 6 as second job, i.e. $t^{\prime} = p_1 + p_6$. Instead, if  the non critical machine is the second machine, then $t^{\prime}$ must be contributed by $p_2$ and $p_5$, i.e. $t^{\prime} = p_2 + p_5$. To show this, first notice that the assignment $t^{\prime} = p_2 + p_6$, which also implies $p^{\prime} = p_7$ as well as having job 5 with job 3 and 4, cannot occur, as it would prevent the critical machine from scheduling three jobs. Also, job 7 cannot be scheduled as second job on such non critical machine. 
Eventually, a third possibility is to have $t^{\prime} = p_3 + p_4$. Solving model \eqref{eq:ObjOptB}--\eqref{eq:TheVarB} with the constraints related to the three possible assignments of $t^{\prime}$ provides the following upper bounds on ratio $\frac{C_{m}^{LPT}}{C_{m}^*}$:
\begin{align}
& t^{\prime} = p_1 + p_6 \implies \frac{1}{opt} = \frac{15}{13} \left(< \frac{7}{6} \right); \nonumber \\
& t^{\prime} = p_2 + p_5  \implies \frac{1}{opt} = \frac{7}{6}; \nonumber \\
& t^{\prime} = p_3 + p_4  \implies \frac{1}{opt} = \frac{7}{6}. \nonumber
\end{align}
Putting together all the stated results, we get the claim.
\end{proof}

\bibliographystyle{elsarticle-harv} 
 \bibliography{ref88}
 \label{sec:bib}
 
\end{document}